\documentclass{article} % [12pt, a4]
\usepackage{graphicx}

\usepackage{tikz}
\usetikzlibrary{arrows}

\usepackage{multicol}
\usepackage{mathtools}
\usepackage[noend]{algpseudocode}
\usepackage{epsfig}
\usepackage{braket}
\usepackage{rotating}
\usepackage{amsmath}
\usepackage{algorithm}
\usepackage{xspace}
\usepackage{comment}
\usepackage{hyperref}

\usepackage{amssymb}
\usepackage{amsthm}
\usepackage{booktabs}
\usepackage{multirow}
\usepackage{tablefootnote}

\usepackage{mathtools}
\DeclarePairedDelimiter\ceil{\lceil}{\rceil}

\newtheorem{theorem}{Theorem}[section]
\newtheorem{lemma}[theorem]{Lemma}
\newtheorem{proposition}[theorem]{Proposition}

\newtheorem{definition}[theorem]{Definition}

\title{On the Quantum Chromatic Numbers of Small Graphs}
\author{Olivier Lalonde\\ \\
 DIRO, Universit\'e de Montr\'eal,\\
 2920, chemin de la Tour, Montr\'eal, Qu\'ebec, Canada H3T 1N8\\
 {\tt olivier.lalonde.1@umontreal.ca} \\
}

\date{\today}

\begin{document}
\maketitle

\begin{abstract}
We make two contributions pertaining to the study of the quantum chromatic numbers of small graphs. Firstly, in an elegant paper, Man\v{c}inska and Roberson [\textit{Baltic Journal on Modern Computing}, 4(4), 846-859, 2016] gave an example of a graph $G_{14}$ on 14 vertices with quantum chromatic number 4 and classical chromatic number 5, and conjectured that this is the smallest graph exhibiting a separation between the two parameters. We describe a computer-assisted proof of this conjecture, thereby resolving a longstanding open problem in quantum graph theory. Our second contribution pertains to the study of the rank-$r$ quantum chromatic numbers. While it can now be shown that for every $r$, $\chi_q$ and $\chi^{(r)}_q$ are distinct, few small examples of separations between these parameters are known. We give the smallest known example of such a separation in the form of a graph $G_{21}$ on 21 vertices with $\chi_q(G_{21}) =  \chi^{(2)}_q(G_{21}) = 4$ and $ \xi(G_{21}) = \chi^{(1)}_q(G_{21}) = \chi(G_{21}) = 5$. The previous record was held by a graph $G_{msg}$ on 57 vertices that was first considered in the aforementioned paper of Man\v{c}inska and Roberson and which satisfies $\chi_q(G_{msg}) = 3$ and $\chi^{(1)}_q(G_{msg}) = 4$. In addition, $G_{21}$ provides the first provable separation between the parameters $\chi^{(1)}_q$ and $\chi^{(2)}_q$. We believe that our techniques for constructing $G_{21}$ and lower bounding its orthogonal rank could be of independent interest.

\vspace{0.1in}

\end{abstract}

\section{Introduction}
Let $G$ be a finite simple graph. For a given number of colours $k$, this paper will be centered around the following scenario, which was first studied in full generality by Galliard and Wolf (\cite{GW}) and is called the \textit{$k$-colouring game} on $G$. Two participants, traditionally named Alice and Bob, are physically separated,  respectively given vertices $x$, $y$ of $G$ under the promise that $x$ and $y$ are either equal or adjacent, and then requested to respectively output colours $a,b \in [k]$ that are equal if and only if their inputs were. In a given setting, we are interested in the smallest value of $k$, denoted $k^*$, for which this can be achieved with certainty: clearly, this will also be possible for all $k \geq k^*$. \\

It can be seen that perfect deterministic strategies for the $k$-colouring game on $G$ and classical $k$-colourings of $G$ are equivalent concepts: fixing such a $k$-colouring, Alice and Bob can play the game perfectly by each outputting the colour corresponding to the vertex they were given, and conversely, it is fairly straightforward to see that all perfect deterministic strategies for the game are of this form. We therefore have that $k^* = \chi(G)$ classically. When quantum mechanics is taken into account, the picture changes: as was first shown by Brassard, Cleve and Tapp (\cite{BCT}), though without using our now-standard graph-theoretic terminology, certain graphs have the intriguing property that the sharing of quantum entanglement enables Alice and Bob to play the $k$-colouring game on $G$ perfectly for some values of $k$ that are strictly smaller than $\chi(G)$. Put differently, defining the \textit{quantum chromatic number} of $G$, denoted $\chi_q(G)$, to be the value of $k^*$ in the entangled setting, it can be the case that $\chi_q(G) < \chi(G)$. $G$ will be said to be \textit{quantumly $k$-colourable} if there exists a perfect entangled strategy for the $k$-colouring game on $G$, and we will refer to such a strategy as a \textit{quantum $k$-colouring} of $G$.\\

As could have been inferred from its title, this paper is concerned with the study of the quantum chromatic numbers of small graphs. We make two contributions in this direction: one negative, by giving a tight lower bound on the size of a graph exhibiting a separation between the classical and quantum chromatic numbers, and one positive, by giving an example of a small graph with interesting properties from the standpoint of the study of the quantum chromatic number. Most of our results will be computer-assisted, and most of the corresponding code, which is written in the Julia language, can be found at \url{https://github.com/lalondeo/QuantumColorings}, which itself builds on the Koala library (\url{https://github.com/lalondeo/Koala.jl}). \\

Firstly, while it was first shown in \cite{BCT} that the classical and quantum chromatic numbers are distinct in general (though without giving a concrete example of a graph with this property: their work was made explicit some time later by Galliard, Tapp and Wolf (\cite{GTW})), the smallest graph arising from this line of work that exhibits the desired separation, which is due to Avis, Hasegawa, Kikuchi and Sasaki (\cite{avis}), contains more than a thousand vertices. This means that the corresponding colouring games are not very well suited for an experimental demonstration of the existence of quantum entanglement. In addition to initiating the formal study of the quantum chromatic number, Cameron, Montanaro, Newman, Severini and Winter (\cite{cameron06}) came up with a much smaller separation between the classical and chromatic numbers in the form of a graph $G_{18}$ on 18 vertices with $\chi(G_{18}) = 5$ and $\chi_q(G_{18}) = 4$. The corresponding quantum $4$-colouring of $G_{18}$ is obtained by invoking proposition \ref{quaternion} below together with the fact that $G_{18}$ admits an orthogonal representation in $\mathbb{R}^4$, by construction. This was later improved upon by Man\v{c}inska and Roberson (\cite{mancinska}), who gave a graph $G_{14}$ on 14 vertices which otherwise shares the aforementioned properties of $G_{18}$ and conjectured that this separation is smallest possible. In section 3, we will describe a computer-assisted proof of the following slight strengthening of their conjecture:
\begin{theorem} \label{noseparation}
All graphs $G$ with $\chi_q(G) < \chi(G)$ satisfy one of the following:
\begin{enumerate}
\item $|V(G)| \geq 15$
\item $|V(G)| = 14$ and $\chi_q(G) \geq 4$
\end{enumerate}
\end{theorem}
\smallskip
The above result shows that the separation between the quantum and classical chromatic numbers that is given by $G_{14}$ is minimal, both in terms of its size, and among graphs of the same size exhibiting the separation, in terms of its quantum chromatic number. The longstanding open problem of determining the smallest graph exhibiting a separation between the quantum and classical chromatic numbers, which seems to have appeared in print for the first time in the work of \cite{avis} and was also asked by \cite{cameron06}, is thereby completely solved. The idea of the proof of theorem \ref{noseparation} is to exhaustively enumerate a certain class of graphs with the property that if a counterexample to the theorem existed, then that class would also contain a counterexample. This enumeration was achieved using an algorithm based on the \texttt{geng} program of the NAUTY library of McKay (\cite{mckay}), to be described in subsection 3.1. With the resulting list in hand, we then ran every graph $G$ in it through a specialised semidefinite hierarchy due to Russell (\cite{russell}) to attempt to show that the colouring game on $G$ with $\chi(G)-1$ colours has entangled synchronous value strictly less than one, which implies that $G$ is not quantumly $(\chi(G)-1)$-colourable and therefore that $\chi_q(G) = \chi(G)$, as desired.  How this was done is described in subsection 3.2. \\

Our first contribution showed the absence of a separation between the quantum and classical chromatic numbers among a certain class of graphs. By contrast, in section 4, we will provide an example of a small graph $G_{21}$ on 21 vertices exhibiting a more exotic separation between the quantum and classical chromatic numbers that that which is given by $G_{14}$ and $G_{18}$. Prior to this work, all examples of small graphs exhibiting this separation relied on proposition \ref{quaternion}, so that for the corresponding graphs $G$, it is always the case that there exists a quantum $\chi_q(G)$-colouring of $G$ with corresponding measurement operators all of rank one. Equivalently, defining the \textit{rank-$r$ quantum chromatic number} (\cite{cameron06}) of $G$, denoted $\chi^{(r)}_q(G)$, to be the smallest value of $k$ for which $G$ admits a quantum $k$-colouring with all measurement operators of rank exactly $r$ (see theorem \ref{structurecameron} for why this is a useful definition to make), we have that for these graphs, $\chi^{(1)}_q(G) = \chi_q(G)$. \cite{cameron06} asked if this last equation holds for all graphs. Although this question was seemingly never addressed directly in the subsequent literature, it can now be established that this is not the case, and more strongly, recent results of Harris (\cite{harris}) can be seen to imply that for every $r$, there exists a graph $G$ with $\chi_q(G) = 3$ and $\chi^{(r)}_q(G) > 3$. Indeed, as was realised by \cite{cameron06}, the parameter $\chi^{(r)}_q$ is computable in principle, while \cite{harris} showed that the problem of determining whether a given graph is quantumly 3-colourable is undecidable in general, by appealing to a result of Slofstra (\cite{slofstra}). The work of \cite{harris} can be leveraged in principle to produce an explicit separation between $\chi_q$ and $\chi^{(r)}_q$ for any given choice of $r$, but the resulting graphs will most likely be formidably large. Prior to this work, the only reasonably small graph known to exhibit a separation of this kind came from the work of \cite{mancinska} and is a graph $G_{msg}$ on 57 vertices with $\chi_q(G_{msg}) = 3$ and $\chi^{(1)}_q(G_{msg}) = 4$. $G_{msg}$ is obtained by applying Karp's classical reduction from 3-SAT to 3-COL to the system of equations defining the magic square game of Mermin (\cite{mermin}) and Peres (\cite{peres}), and the fact that $\chi^{(1)}_q(G_{msg}) = 4$ follows from the fact that $\xi(G_{msg}) = \chi(G_{msg}) = 4$, as shown by \cite{mancinska}, combined with the first part of proposition \ref{quaternion}. The measurement operators in the quantum 3-colouring of $G_{msg}$ corresponding to the standard perfect quantum strategy for the magic square game are not all of the same rank, but by using the averaging trick of \cite{cameron06}, they can all be made to be of rank 4, so that $\chi^{(4)}_q(G_{msg}) = 3$. It may be that $G_{msg}$ admits a rank-3 (or even rank-2) quantum 3-colouring, but we do not know of one. In addition to being much smaller than $G_{msg}$, our graph $G_{21}$ could be shown to satisfy $\chi_q(G_{21}) = \chi^{(2)}_q(G_{21}) = 4 $ and $\chi^{(1)}_q(G_{21}) = \chi(G_{21}) = 5$. Therefore, $G_{21}$ is the smallest graph known to exhibit a separation between $\chi_q$ and $\chi^{(1)}_q$, and in particular, the smallest known separation between the classical and quantum chromatic numbers such that no quantum colouring achieving the separation can be built by appealing to proposition \ref{quaternion}. Moreover, this gives the first proof of the following result:
\begin{theorem} \label{separation21}
The parameters $\chi^{(1)}_q$ and $\chi^{(2)}_q$ are distinct.
\end{theorem}
As described in subsection 4.1, the graph $G_{21}$ is obtained as the orthogonality graph of the vector clumps listed in appendix A, which theorem \ref{clumps} shows how to turn into a rank-two quantum 4-colouring of $G_{21}$. It is easy to show that $\chi(G_{21}) = 5$ using a computer, so that to show that $\chi^{(1)}_q(G_{21}) = 5$, it is enough to establish that $\xi(G_{21}) > 4$, as per proposition \ref{quaternion}. Our computer-assisted proof of this fact, which is described in subsection 4.2, relies on a new branch-and-bound-like algorithm predicated on a generalisation of the square-free criterion of Arends, Ouaknine and Wampler (\cite{arends}) and on the graph parameter $\xi_{SDP}$ of Paulsen, Severini, Stahlke, Todorov and Winter (\cite{chi-co}), which is a strengthening of sorts of the classical $\vartheta$ number of Lovász  (\cite{lovasz}).

\section{Preliminaries}
The graphs under consideration in this paper are all finite and simple, and the size of a graph will always be taken to mean the number of vertices in it. We will frequently use the shorthand $[n]$ to mean $\{1,...,n\}$. \\

Let $G$ be a graph. For a given $k \in \mathbb{N}$, a classical \textit{$k$-colouring} of $G$ is an assignment $\{c_v\}_{v \in V(G)}$ of elements of $[k]$, which we traditionally refer to as colours, to the vertices of $G$ such that $c_u \neq c_v$ for all $(u,v) \in E(G)$. The smallest value of $k$ for which a $k$-colouring of $G$ exists is called the \textit{chromatic number} of $G$, denoted $\chi(G)$. A \textit{clique} is a set of pairwise adjacent vertices of $G$: the size of the largest clique of $G$ is called the \textit{clique number} of $G$, denoted $\omega(G)$. An \textit{independent set} is a set of pairwise nonadjacent vertices of $G$.  A \textit{$k$-dimensional orthogonal representation} of $G$ is an assignment of unit vectors $\{\ket{\psi}_v\}_{v \in V(G)}$ in $\mathbb{C}^k$ to the vertices of $G$ such that, for all $(u,v) \in E(G)$, $\braket{\psi_u|\psi_v} = 0$: the smallest value of $k$ for which $G$ admits a $k$-dimensional orthogonal representation is called the \textit{orthogonal rank} of $G$, denoted $\xi(G)$. The following is standard and simple to see:
\begin{proposition} \label{sandwich}
For all graphs $G$, it holds that
\[\omega(G) \leq \xi(G) \leq \chi(G)\]
\end{proposition}
\bigskip
Following Paulsen and Todorov (\cite{operatorsystems}), each entanglement model $t \in \{q, qa, qc\}$ gives rise to a quantum chromatic number $\chi_t$, defined to be the smallest value of $k$ for which Alice and Bob can win the $k$-colouring game on a given graph with certainty when given access to correlations in the corresponding correlation set. We easily have, for every graph $G$:
\[\chi_{qc}(G) \leq \chi_{qa}(G) \leq \chi_q(G) \leq \chi(G)\]
In line with the literature, we always mean $\chi_q$ when we speak of `the' quantum chromatic number, although we are mainly concerned with small graphs in this paper, for which the three quantum chromatic numbers are expected to always coincide. Though it had already been defined in a passing remark in \cite{avis} (which attributes its definition to Patrick Hayden), the extensive study of the parameter $\chi_q$ was first undertaken by \cite{cameron06}, while the specific study of $\chi_{qc}$ was first carried out in \cite{chi-co}. The three quantum chromatic numbers were recently proven to all be uncomputable and distinct by \cite{harris}, building on the work of Ji (\cite{ji}) and on recent breakthrough results in the theory of nonlocal games, namely those of Slofstra (\cite{slofstra}, \cite{slofstra2}) and of Ji, Natarajan, Vidick, Wright and Yuen (\cite{MIPRE}). \\

In the finite-dimensional case, we have the following convenient structure theorem due to \cite{cameron06}, which provides a simpler way to specify a quantum $k$-colouring of a given graph: 
\smallskip
\begin{theorem}[\cite{cameron06}] \label{structurecameron}
Given a graph $G$ and $k \in \mathbb{N}$, $G$ is quantumly $k$-colourable if and only if, for some finite-dimensional complex Hilbert space $\mathcal{H}$, there exists an assignment of a projective measurement $\{E^v_c\}_{c \in [k]}$ on $\mathcal{H}$ to every $v \in V(G)$ in such a way that for every $(u,v) \in E(G)$ and for every $c \in [k]$, it holds that $E^u_c E^v_c = 0$. This assignment corresponds to the following quantum $k$-colouring of $G$: Alice and Bob share the standard maximally entangled state on $\mathcal{H}_A \otimes \mathcal{H}_B$; on input $x$, Alice measures her part of the state with respect to $\{E^x_a\}_{a \in [k]}$; on input $y$, Bob measures his part of the state with respect to $\{\overline{E}^y_b\}_{b \in [k]}$; and they both output the results. Furthermore, a perfect entangled strategy for the $k$-colouring game on $G$ can be assumed to be of the above form without loss of generality, i.e. without increasing the local dimension of the shared entangled state. Moreover, the above projectors can all be assumed to be of the same rank $r$, though at the cost of potentially increasing the local dimension of the shared entangled state.
\end{theorem}
For a given $r$, following \cite{cameron06}, the rank-$r$ chromatic number $\chi^{(r)}_q(G)$ is defined to be the smallest $k$ such that $G$ admits an assignment of the form above with projectors all of rank exactly $r$. It is not hard to see that $\chi_q^{(r+1)}(G) \leq \chi^{(r)}_q(G)$ for all $r$, and by the last part of the previous theorem, it holds that $\chi_q(G) = \min_r \chi_q^{(r)}(G)$. As mentioned in the introduction, computability considerations imply that for every $r$, there exists a graph $G$ with $\chi_q(G) = 3$ and $\chi_q^{(r)}(G) > 3$, but although this seems very likely to be the case, no proof is known that for every $r$, there exists a graph $G$ with $\chi_q^{(r+1)}(G) < \chi_q^{(r)}(G)$. We will later describe the first proof of this for the $r=1$ case. \\

The following proposition collects some results of \cite{cameron06} about the relationships between orthogonal representations, rank-one colourings and classical colourings. While the first part of the proposition is quite simple to prove, the second part requires considerably more ingenuity and is based on a construction involving quaternions and octonions. Prior to this work, this construction was the basis for all known instances of small separations between the quantum and classical chromatic numbers. 
\begin{proposition}[\cite{cameron06}] \label{quaternion}
For all graphs $G$, it holds that
\[\xi(G) \leq \chi_q^{(1)}(G) \leq \chi(G)\]
Also, if, for some $k \in \{2,4,8\}$, $G$ admits an orthogonal representation in $\mathbb{R}^k$, then $\chi_q^{(1)}(G) \leq k$. 
\end{proposition}
\cite{cameron06} asked whether $\xi$ and $\chi_q^{(1)}$ coincide for all graphs. This was shown not to be the case by Scarpa and Severini (\cite{scarpa}). \cite{mancinska} gave a stronger version of this separation in the form of a graph $G_{13}$ on 13 vertices with $\xi(G_{13}) = 3$ and $\chi_q(G_{13}) = \chi^{(1)}_q(G_{13}) = \chi(G_{13}) = 4$. \\

We also record the following standard but useful result:
\begin{proposition} \label{nottwo}
A graph $G$ is such that either $\xi(G) = 2$ or $\chi_{qc}(G) = 2$ if and only if $G$ is bipartite, or equivalently, if $\chi(G) = 2$.
\end{proposition}
\bigskip
In our computations, the real-valued graph parameter $\xi_{SDP}$ of \cite{chi-co} will be very useful to us. Despite its being NP-hard to compute in general, for the small graphs that came under scrutiny in this work, $\xi_{SDP}$ was found to be just as efficiently computable as the $\vartheta$ number while providing stronger bounds. Given a graph $G$, $\xi_{SDP}(G)$ is defined by the following semidefinite program, where $|V(G)| = n$ and where $V(G)$ is being identified with $[n]$: \\
\begin{alignat}{2} 
   \xi_{SDP}(G) =  \quad \quad \quad \min\   & M_{1,1} &\\
   \mbox{s.t. } & M_{i,j} \geq 0 &  i,j \in [n+1]\\
		& M_{1,i+1} = M_{i+1,i+1} = 1,  & i \in [n]\\
     	& \sum_{i \in S} M_{i+1,j+1} \leq 1 & \mbox{ $S$ a clique of $G$, } j \in [n]\\
     	& M_{1,1} + \sum_{i \in S, j \in T} M_{i+1,j+1} \geq |S| + |T| & \mbox{ $S$, $T$ cliques of $G$} \\
     	& M \in \mathbb{S}^{n+1}
\end{alignat}
\smallskip
The following proposition shows that $\xi_{SDP}$ yields bounds on the same parameters of interest as the $\vartheta$ number does, and justifies our thinking of $\xi_{SDP}$ as a strengthening of the $\vartheta$ number. The only part of the proposition that appears to be new is the simple (but important for us) observation that $\xi_{SDP}$ lower bounds the orthogonal rank.\\
\begin{proposition} \label{xisdp}
For all graphs $G$, it holds that
\[\omega(G) \leq \overline{\vartheta}(G) \leq \xi_{SDP}(G) \leq \xi(G), \chi_{qc}(G)\]
\end{proposition}
\begin{proof}
The first inequality is a standard result, and the second is shown in \cite{chi-co}. The last two follow from the fact that $\xi_{SDP}(G) \leq \xi_{tr}(G) \leq \chi_{qc}(G)$ and $\xi_{tr}(G) \leq \xi_f(G) \leq \xi(G)$ (see \cite{chi-co} for the definitions of the corresponding parameters).
\end{proof}
We restricted ourselves to maximal cliques of $G$ in the program (1)-(6) in order to reduce the number of constraints and therefore the solving time, which was generally not found to affect the objective value of the program, and anyway, can only decrease it, so that we still get lower bounds in this way. Also, since the graph parameters that we are directly concerned with are integer valued, only the value of $\ceil{\xi_{SDP}}$ is of interest to us. In order to show that $k \leq \ceil{\xi_{SDP}}$ for some $k$, we look for a feasible solution for the dual of the program (1)-(6) with corresponding objective value at least $k - 1 + \epsilon$, for a fixed small positive value of $\epsilon$. In our implementation, this dual is solved using the \texttt{COSMO} solver of Garstka, Cannon and Goulart (\cite{COSMO}). \\

Finally, we recall that for a given graph $G$, the graph operation called \textit{vertex identification} with respect to two distinct vertices $u,v$ of $G$ amounts to deleting the vertex $u$ and adding the edge $(v,w)$ to the resulting graph for every $w \notin \{u,v\}$ such that $(u,w) \in E(G)$. 

\section{A computer-assisted proof of the minimality of $G_{14}$}

This section describes a computer-assisted proof of theorem \ref{noseparation}. Our proof technique is based on the following simple observation: if $G$ is a graph exhibiting a separation between the quantum and classical chromatic numbers and if $H$ is a subgraph of $G$ with $\chi(G) = \chi(H)$, then $H$ necessarily also exhibits the desired separation, since quantum $k$-colourability is a hereditary property. Defining a graph $G$ to be \textit{edge-critical} if it contains no isolated vertices and if, for every proper subgaph $H$ of $G$, we have $\chi(H) < \chi(G)$, or, equivalently, the removal of any vertex or edge from $G$ causes its chromatic number to decrease, it is then apparent that if there existed a counterexample to theorem \ref{noseparation}, an edge-critical counterexample would necessarily exist also. Defining a graph $G$ to be \textit{edge-$k$-critical} if it is edge-critical and $\chi(G) = k$, we will list all edge-$k$-critical graphs on $n$ vertices for $4 \leq k \leq n \leq 13$ and for $k=4$, $n=14$. Theorem \ref{noseparation} is then equivalent to the statement that none of the resulting graphs exhibits a separation between the classical and quantum chromatic numbers: note that the edge-3-critical graphs can be omitted in view of proposition \ref{nottwo}. The proof of the theorem is then completed by using a semidefinite hierarchy due to \cite{russell} to show that for all the resulting graphs $G$, the commuting operators value of the $(\chi(G)-1)$-colouring game on $G$ is strictly less than one. It should be noted that this also rules out the possibility that $\chi_{qc}(G) < \chi(G)$, so that theorem \ref{noseparation} goes through for $\chi_{qc}$ as well. The algorithm that we used to enumerate edge-$k$-critical graphs on a given number of vertices is the subject of subsection 3.1, and subsection 3.2 explains in more detail how we go about proving a lower bound on the quantum chromatic number of a given graph.

\subsection{Enumerating edge-critical graphs}
Given $4 \leq k \leq n$, we describe our algorithm for generating all edge-$k$-critical graphs on $n$ vertices. Our approach is built upon the \texttt{geng} program of the NAUTY package of McKay (\cite{mckay}), which enumerates all graphs on a given number of vertices satisfying certain properties exactly once up to isomorphism. \texttt{geng} allows the user to specify certain such properties out of the box, such as connectedness or minimum degree, and also allows for the implementation of custom filters through the functions \texttt{prune} and \texttt{preprune}. Algorithm \ref{geng} describes the role of these functions in the course of the algorithm. \texttt{preprune} is called at an earlier stage and more often than \texttt{prune}, while \texttt{prune} is called at most once on every graph. 
\begin{algorithm} 
\caption{A cartoon depiction of the inner workings of \texttt{geng}. \texttt{shouldskip} makes sure that only one isomorphic copy of $G$ is examined and enforces the restrictions specified by the user. \texttt{preprune} and \texttt{prune} are to be implemented by the user.  } \label{geng}
\begin{algorithmic}
\Function{\texttt{enumerate}}{$G$, $n$}
	\If{\texttt{preprune}($G$), \texttt{shouldskip}($G$) and \texttt{prune}($G$) are all false (tested in this order)}
		\If{$G$ has $n$ vertices}
			\State Record $G$
		\Else
			\ForAll{graphs $G'$ obtained from $G$ by adding a vertex to it and by adding edges between that vertex and the other vertices}
				\State \texttt{enumerate}($G'$, $n$)
			\EndFor
		\EndIf
	\EndIf
\EndFunction
\State \texttt{enumerate}(the graph on one vertex, $n$)
\end{algorithmic}
\end{algorithm}

In our implementation, we specify to \texttt{geng} that all graphs on $n$ vertices under consideration are to have minimum degree at least $k-1$. This is a correct restriction because if a vertex of a given graph $G$ has degree at most $\chi(G)-2$, it can be seen that deleting that vertex from the graph will not affect its chromatic number, which shows that the graph is not edge-critical. We now turn to explaining how our \texttt{preprune} and \texttt{prune} functions operate. In all that follows, we use the naive backtracking algorithm for testing for $k$-colourability and for listing colourings, with the slight twist that we begin by finding a large enough maximal clique $v_1, \ldots, v_l$ in the graph and we force the vertex $v_i$ to be coloured with colour $i$ for every $i$ to reduce the size of the search space. \\

The idea of our implementation is to do as much precomputation as possible on the graphs on $n-1$ vertices so as to minimise the work that needs to be done on the graphs on $n$ vertices, which are far more numerous. In our case, \texttt{preprune} always returns `false' unless it is called on a graph on $n$ vertices. Given a graph $G$ on $n-1$ vertices, $\texttt{prune}$ begins by making sure that $\chi(G) = k-1$. If this does not hold, it can be seen that no graph $G'$ obtained from $G$ by adding a vertex can be edge-$k$-critical and so $G$ may be pruned. Otherwise, we do some precomputation on $G$ to be able to run heuristic tests on any extension $G'$ of $G$ to rule out edge-$k$-criticality quickly in most cases. We compute the following data in \texttt{prune}: 
\begin{enumerate}
\item We list a number of distinct independent sets of size at most 2 $I_1, I_2, \ldots, I_n$, each with the property that there exists a $(k-1)$-colouring of $G$ such that, for some colour, the vertices coloured with that colour are precisely the vertices contained in that independent set. This is done by simply listing all $(k-1)$-colourings of $G$ and examining each of them. 
\item We pick distinct edges $(u_1, v_1), \ldots, (u_n, v_n)$ of $G$ arbitrarily. For every such edge $(u,v)$, we list all assignments of colours in $[k-1]$ to the vertices of $G$ such that any two adjacent vertices are assigned different colours except for $u$ and $v$, which are assigned the same colour. This is achieved by contracting $(u,v)$ and by listing the $(k-1)$-colourings of the resulting graph. In the event that too many colourings were found, which would harm performance, we give up on the edge and choose another one. The edges are then sorted by increasing number of colourings found. In our implementation, $n=4$ was found to yield the best performance.
\end{enumerate}
Then, in \texttt{preprune}, given a graph $G'$ on $n$ vertices which extends $G$, we run the following tests in the given order:
\begin{enumerate}
\item Heuristic $(k-1)$-colouring test: for every independent set that was previously computed for $G$, check if some vertex in that set is adjacent to the last vertex of $G'$. If not, by construction, $G'$ is $(k-1)$-colourable and may be pruned.
\item Heuristic edge-criticality test: for every edge $(u,v)$ picked previously, check if at least one of the previously computed assignment is such that, for some colour $c \in [k-1]$, $c$ was not assigned to any of the neighbors of the last vertex. If not, $G'$ is not edge-$k$-critical and may be pruned. To see why, suppose that $G'$ is edge-$k$-critical and consider the graph $G''$ obtained from $G'$ by removing the edge $(u,v)$. By hypothesis, $G''$ is $(k-1)$-colourable; and furthermore, any $(k-1)$-colouring of $G''$ must assign the same colour to $u$ and $v$ because, otherwise, the colouring would also be a valid $(k-1)$-colouring for $G'$, which is supposed to have chromatic number $k$. Therefore, if this test does not pass, either $\chi(G') = k-1$ or $\chi(G'') = k$. In either case, we can conclude that $G'$ is not edge-$k$-critical. 
\end{enumerate} 
Finally, in \texttt{prune}, a full-blown edge-$k$-criticality test is run. Little regard to efficiency is paid at this point because of how powerful the previous two heuristic tests are. 

\smallskip

Our implementation of the above algorithm, which is written in C, can be found over at \url{https://github.com/lalondeo/gencrit}. Brendan McKay gracefully agreed to host the lists we produced, which can now be found over at \url{https://users.cecs.anu.edu.au/~bdm/data/graphs.html}. The computation took roughly one year of CPU time, with the great majority of this time going into enumerating the 4-critical graphs on 14 vertices, and yielded around 13.8 million graphs. The number of edge-$k$-critical graphs on $n$ vertices for every $k$,$n$ is given in table \ref{numberofgraphs}. It should be mentioned that McKay had separately enumerated all edge-4-critical graphs on 13 vertices or less and his lists agree exactly with ours. The others were validated by first checking that every graph in the list is indeed edge-critical as well as by generating a large number of edge-critical graphs at random and by checking that they were all enumerated. It can therefore be asserted with high confidence that the lists we produced are complete. \\

\begin{table}[H] \label{numberofgraphs}
\caption{The number of edge-$k$-critical graphs for every graph size $n$, up to isomorphism. There are 13,778,383 graphs in total. It is interesting that the number of graphs on every diagonal appears to be converging to a definite value. We do not know why that is. }
\begin{tabular}{cl|llllllllll}
\cline{3-12}
\multicolumn{2}{l|}{\multirow{2}{*}{}}        & \multicolumn{10}{c|}{$k$}                                                                 \\
\multicolumn{2}{l|}{}                         & 4       & 5       & 6       & 7     & 8   & 9  & 10 & 11 & 12 & \multicolumn{1}{l|}{13} \\ \hline
\multicolumn{1}{|c}{\multirow{11}{*}{$n$}} & 4  & 1       & -       & -       & -     & -   & -  & -  & -  & -  & -                       \\
\multicolumn{1}{|c}{}                    & 5  & 0       & 1       & -       & -     & -   & -  & -  & -  & -  & -                       \\
\multicolumn{1}{|c}{}                    & 6  & 1       & 0       & 1       & -     & -   & -  & -  & -  & -  & -                       \\
\multicolumn{1}{|c}{}                    & 7  & 2       & 1       & 0       & 1     & -   & -  & -  & -  & -  & -                       \\
\multicolumn{1}{|c}{}                    & 8  & 5       & 2       & 1       & 0     & 1   & -  & -  & -  & -  & -                       \\
\multicolumn{1}{|c}{}                    & 9  & 21      & 21      & 2       & 1     & 0   & 1  & -  & -  & -  & -                       \\
\multicolumn{1}{|c}{}                    & 10 & 150     & 162     & 22      & 2     & 1   & 0  & 1  & -  & -  & -                       \\
\multicolumn{1}{|c}{}                    & 11 & 1,221    & 4,008    & 393     & 22    & 2   & 1  & 0  & 1  & -  & -                       \\
\multicolumn{1}{|c}{}                    & 12 & 14,581   & 147,753  & 17,036   & 395   & 22  & 2  & 1  & 0  & 1  & -                       \\
\multicolumn{1}{|c}{}                    & 13 & 207,969  & 8,311,809 & 1,479,809 & 25,355 & 395 & 22 & 2  & 1  & 0  & 1                       \\
\multicolumn{1}{|c}{}                    & 14 & 3,567,180 & ?       & ?       & ?     & ?   & ?  & ?  & ?  & ?  & ?                       \\ \cline{1-2}
\end{tabular}
\end{table}

\subsection{The pipeline for proving a lower bound on the quantum chromatic number of a given graph}
Having explained how the edge-critical graphs are enumerated, we now turn to describing a procedure for attempting to prove that a given graph does not admit a quantum $k$-colouring for a given value of $k$. We will then systematically apply this to every previously enumerated graph $G$ with $k = \chi(G)-1$ in order to show that $\chi_q(G) = \chi(G)$. We note that we are not the first to approach the problem of lower bounding the quantum chromatic number of a given graph using computation: \cite{mancinska} report that Piovesan and Burgdorf could find an alternative proof of their result that their graph $G_{13}$ is not quantumly 3-colourable using a computer algebra system. \cite{mancinska} remarked that this approach did not work in many other cases, and in particular, was seemingly unable to show that a graph does not admit a quantum $k$-colouring for any $k \geq 4$. This was later proven to be true in general by Helton, Meyer, Paulsen and Satriano (\cite{algebras}). Therefore, while their approach could conceivably have been used to deal with the edge-4-critical graphs, another proof strategy is required to handle the graphs with larger chromatic numbers. \\

Our proof technique is based on a semidefinite hierarchy of the kind that was first put forth independently by Navascués, Pironio and Ac\'{i}n (\cite{NPA1}) and by Doherty, Liang, Toner and Wehner (\cite{NPA2}). While this original hierarchy could have been used directly, in our case, a more efficient alternative exists in the form of the hierarchy of Russell (\cite{russell}), which is specialised to so-called \textit{synchronous} correlations. Given finite input and output sets $X$ and $A$, a correlation $p(a,b|x,y)_{a,b \in A, x,y \in X}$ is said to be synchronous if, for all input pairs $x,y \in X$, the outputs $a,b \in A$ are equal with probability one. By definition, any correlation which wins a colouring game with probability one is necessarily synchronous. Taking $X$ and $A$ to be fixed, for a given entanglement model $t \in \{q, qc\}$, we will write $C^s_t$ to mean the set of synchronous correlations in the entanglement model $t$. Given a nonlocal game $\mathcal{G}$ specified by identical input sets $X$, identical output sets $A$, an input distribution $q_{x,y}$ on $X \times X$ and a predicate $V: X \times X \times A \times A \mapsto \{0,1\}$, with respect to entanglement model $t \in \{q, qc\}$, the \textit{synchronous entangled value} (\cite{syncval}) of $\mathcal{G}$, denoted $\omega^s_t(\mathcal{G})$, is defined by the supremum 
\begin{alignat}{2} 
   \sup \   & \sum_{x, y \in X} q_{x,y} \sum_{a, b \in A} p(a,b|x,y) V(x,y,a,b) \\
   \mbox{s.t. } & p(a,b|x,y) \in C^s_t
\end{alignat}
Similarly to the original hierarchy of \cite{NPA1} and \cite{NPA2}, the hierarchy of \cite{russell} consists in a nonincreasing sequence $C_1 \supseteq C_2 \supseteq ...$ of synchronous correlation sets all containing $C^s_{qc}$ and converging to it in the limit and which are all specified by semidefinite constraints. For every level $i$, we then consider the semidefinite program
\begin{alignat}{2} 
   \max \   & \sum_{x, y \in X} q_{x,y} \sum_{a, b \in A} p(a,b|x,y) V(x,y,a,b) \\
   \mbox{s.t. } & p(a,b|x,y) \in C_i
\end{alignat}
These programs can be optimised efficiently on a computer and their optimal values provide asymptotically tight upper bounds on $\omega^s_{qc}(\mathcal{G})$, though the complexity of computing the optimal value of the $i$-th program is exponential in $i$, meaning that unless $\mathcal{G}$ is very small, only the first few upper bounds can feasibly be computed. In our case, given a graph $G$ and $k \in \mathbb{N}$, writing $\mathcal{G}_{G,k}$ to denote the $k$-colouring game on $G$ with the input distribution taken to be uniform over the legal inputs, our approach to showing that $G$ is not quantumly $k$-colourable is to try to prove the stronger statement that $\omega^s_{qc}(\mathcal{G}_{G,k}) < 1$ by means of the above hierarchy. It can be seen that this holds if and only if $k < \chi_{qc}(G)$, and as was mentioned in the preliminaries, it is now known that it may be that $\chi_{qc}(G) < \chi_q(G)$, in which case the approach we are outlining would be powerless at showing a tight lower bound on $\chi_q$. There is no way around this, at least in the $k=3$ case: it follows from computability considerations that the set of graphs that are not quantumly 3-colourable is not recursively enumerable, and hence that there exists no computational procedure for proving that a given graph is not quantumly 3-colourable that will systematically eventually succeed if this does hold. Since the graphs presently under consideration are quite small, our hope is that the approach we described will work nevertheless. \\

We now turn to describing in more detail how the hierarchy of \cite{russell} is built. Letting $\mathcal{A}$ be a $C^*$-algebra, a state $\tau: \mathcal{A} \mapsto \mathbb{R}$ is said to be \textit{tracial} if $\tau(AB) = \tau(BA)$ for all $A,B \in \mathcal{A}$. The linchpin of the hierarchy of \cite{russell} is the following lemma of \cite{chi-co}, which can be seen to parallel theorem \ref{structurecameron}:
\begin{lemma}[\cite{chi-co}]
The correlations in $C^s_{qc}$ are precisely those for which there exists a $C^*$-algebra $\mathcal{A}$, a tracial state $\tau: \mathcal{A} \mapsto \mathbb{R}$ and an assignment of projective measurements to the elements of $X$ $\{E^x_a\}_{x \in X, a \in A}$  such that, for every $x,y \in X$, $a,b \in A$, 
\[p(a,b|x,y) = \tau(E^x_a E^{y}_{b})\]
\end{lemma}
\smallskip
The way the semidefinite hierarchy of \cite{russell} is derived from this lemma is very similar to the way the original semidefinite hierarchy of \cite{NPA1} and \cite{NPA2} is defined. Taking the alphabet $\Sigma$ to be the collection of wooden symbols $\{E^x_a\}_{x \in X, a \in A}$, which we think of as being placeholders for projective measurements on a generic $C^*$-algebra, and given a finite set $S \subseteq \Sigma^*$ containing all words of length at most one, which we think of as being monomials over the $E^x_a$, we look at the space of so-called pseudo-states $\psi$, which are functions $S \times S \mapsto \mathbb{R}$ satisfying certain properties that would have to be satisfied if, for some $C^*$-algebra $\mathcal{A}$, there existed projective measurements $E^x_a$ in $\mathcal{A}$ and a tracial state $\tau$ on $\mathcal{A}$ such that $\psi(x^R, y) = \tau(x^\dagger y)$ for every $x, y \in S$. The corresponding correlation set $C_S$ is then defined to be the set of correlations $p(a,b|x,y)$ such that, for some such pseudo-state $\psi$, $p(a,b|x,y) = \psi(E^x_a, E^y_b)$ for all $x,y,a,b$: by construction and by the previous lemma, we have that $C^s_{qc} \subseteq C_S$. $\psi$ can be viewed as a square matrix of size $|S|$, which can be constrained to be positive semidefinite, and the algebraic constraints that are imposed on it are linear in its entries, so that membership in $C_S$ is indeed specified by a semidefinite program. The precise algebraic constraints that are being imposed are somewhat tedious to spell out, and the reader is referred to \cite{russell} for further details. Tying back to our previous high-level presentation, for every $i \in \mathbb{N}$, setting $S_i$ to be the collection of all words on $\Sigma$ of length at most $i$, we take $C_i = C_{S_i}$. For our purposes, however, this presentation is too coarse-grained. On the one hand, while the program (9)-(10) corresponding to the set $C_1$ is very small and can be optimised very quickly in all cases, for a significant proportion of our graphs $G$, it is too weak to rule out the existence of a quantum $(\chi(G)-1)$-colouring. In fact, corollaries 15 and 16 of Cubitt, Man\v{c}inska, Roberson, Severini, Stahlke and Winter (\cite{cubitt}) can be seen to imply that the test corresponding to $C_1$ will succeed in showing that a given graph $G$ is not quantumly $k$-colourable if and only if $k < \overline{\vartheta}^+(G)$, where $\vartheta^+$ is Szegedy's (\cite{szegedy}) strengthening of the Lovász $\vartheta$ number. Since $\overline{\vartheta}^+(G) \leq \xi_{SDP}(G)$ for all graphs, this test is subsumed by the first step of our pipeline. On the other hand, while the test corresponding to the set $C_2$ was found to be extremely strong, the resulting semidefinite program is quite large and takes a fair amount of time to solve for the larger graphs with large chromatic numbers. A simple workaround is to run the test with carefully chosen subsets of $S_2$ instead, which reduces the amount of computation required by a fair amount. \\

The full pipeline that was used for proving that a given graph $G$ satisfies $\chi_{q}(G) = \chi(G)$ is the following, where a particular ordering of $V(G)$ is chosen arbitrarily, where $k = \chi(G)-1$ and where, given a choice of monomials $S$, the program (9)-(10) is being optimised for the game $\mathcal{G}_{G,k}$ with correlation set $C_S$. We ran every test in the list sequentially until one could prove the nonexistence of a quantum $k$-colouring of $G$.
\begin{enumerate}
\item Try to prove that $\ceil{\xi_{SDP}(G)} = \chi(G)$. 
\item Run the test given by the hierarchy corresponding to the set
\[S = \{E^x_c E^y_c \mid x,y \in V(G), c\in [k], x < y\}\]
\item Run the test given by the hierarchy corresponding to the set
\[S' = \{E^x_c E^y_c \mid x,y \in V(G), c \in [k]\}\]
\item Run the test given by the hierarchy corresponding to a subset $S''$ of $C_2$ chosen by randomly discarding every monomial of the form $E^x_a E^y_b$ where $x > y$ with probability 50\%. This random choice is seeded so as to be reproducible. 
\end{enumerate}
In our implementation, the resulting semidefinite programs were all solved using the \texttt{COSMO} solver of Garstka, Cannon and Goulart (\cite{COSMO}). A certificate for the nonexistence of a quantum $k$-colouring was generated by the solver and was validated independently in exact arithmetic. Running the above pipeline on every graph took about 17 weeks of CPU time in total, with more detailed statistics being given by table \ref{latable}. The individual results for every graph can be found at \url{https://www.dropbox.com/scl/fi/r0yqvlltnrzamlzwz9xxq/critical_graphs_results.tar.gz?rlkey=c5bsskychxvewwv0o4ouk6bth&dl=0}. The experiment was run on several different machines of slightly varying computing power, so that the times given are indicative only. Also, a time limit was put on each test, so it may be that some graphs that should not have passed a test did because the solver could not produce a proof of the nonexistence of a quantum $(\chi(G)-1)$-colouring in time. 

\begin{table}[H] \label{latable}
\centering
\caption{Performance of each step of the above battery of tests}
\begin{tabular}{|llll|l}
\cline{1-4}
Test                    & \# of graphs filtered &  \# of remaining graphs & Average CPU time spent (in seconds) &  \\ \cline{1-4}
\multicolumn{1}{|l|}{1} & 10,728,817           & 3,049,566             & 0.015                               &  \\
\multicolumn{1}{|l|}{2} & 3,048,050            & 1,516                 & 3.3                                 &  \\
\multicolumn{1}{|l|}{3} & 1,510                & 6                     & 27.6                                &  \\
\multicolumn{1}{|l|}{4} & 6                    & 0                     & 48.0                                &  \\ \cline{1-4}
\end{tabular}
\end{table}

We see that the $\xi_{SDP}$ test prefiltered most graphs extremely quickly, whereas the second test, while much slower, could rule out a very significant proportion of the remaining graphs. All remaining graphs but six could then be ruled out by the third test. We note that, based on the data of table \ref{latable}, it may seem like the fourth test is not significantly more expensive than the third and hence that the third test could have been done without, but this is because the graphs the fourth test was run on all have chromatic number 4, and the fourth test can enormously slower than the third when applied to graphs with larger chromatic numbers. Interestingly, out of the six graphs that could not be ruled out by the first three tests, only one is on 13 vertices and turns out the be the only edge-4-critical subgraph of the graph $G_{13}$ of \cite{mancinska}, and the other five are on 14 vertices and were all found to be edge-4-critical subgraphs of graphs obtained by cloning a vertex in $G_{13}$. Our pipeline could recover their corollary 1 (which states that $G_{13}$ is not quantumly 3-colourable) in a bit less than two minutes, with 40 seconds being spent on test 4. \\

Since all the previously enumerated graphs could be shown to have equal classical and quantum chromatic numbers, the proof of theorem \ref{noseparation} is complete.

\section{The case of the graph $G_{21}$}
In this section, we give a novel mechanism for constructing quantum colourings, based on the notion of a vector clump which we introduce. This construction will then be put to use by exhibiting a graph $G_{21}$ with $\chi_q(G_{21}) = \chi^{(2)}_q(G_{21}) = 4$ and $\chi(G_{21}) = 5$. This is the contents of subsection 4.1. In subsection 4.2, we will describe a computer-assisted proof of the fact that $\xi(G_{21}) > 4$. Since $\chi(G_{21}) = 5$, this, together with proposition \ref{quaternion}, implies that $\xi(G_{21}) = \chi^{(1)}_q(G_{21}) = 5$. 

\subsection{The clump rank of a graph and $G_{21}$}
We begin by making the following definition:
\begin{definition}
A ($r$,$k$)-\textit{vector clump} is a collection of unit vectors $\{\ket{\psi_{i,j}}\}_{i \in [r], j \in [k]}$ in $\mathbb{C}^{rk}$ which are all pairwise orthogonal. Two ($r$,$k$)-clumps $\{\ket{\psi_{i,j}}\}$ and $\{\ket{\psi'_{i,j}}\}$ are said to be \textit{orthogonal} if, for every $i, i' \in [r]$, it holds that
\[\sum_{j=1}^k \braket{\psi_{i,j}|\psi'_{i',j}} = 0\]
\end{definition}
By analogy with the notion of an orthogonal representation, given a graph $G$, an assignment of $(r,k)$-vector clumps to the vertices of $G$ $\{\ket{\psi^v_{i,j}}\}_{v \in V(G), i \in [r], j \in [k]}$ such that for all $(u,v) \in E(G)$, $\{\ket{\psi^u_{i,j}}\}$ and $\{\ket{\psi^v_{i,j}}\}$ are orthogonal will be called a \textit{($r$,$k$)-clump representation} of $G$. For a given $r$, the \textit{rank-$r$ clump rank} of $G$, denoted $\xi^{(r)}_c(G)$, is defined to be the smallest $k$ such that $G$ admits a ($r$,$k$)-clump representation. It may not be immediately obvious that these parameters are well-defined, but they are, and we even have:
\begin{proposition}
For all graphs $G$ and all ranks $r$, it holds that
\[\xi^{(r)}_c(G) \leq \chi^{(r)}_q(G)\]
\end{proposition}
\begin{proof}
Let $k = \chi^{(r)}_q(G)$ and let $\{E^v_c\}_{v \in V(G), c \in [k]}$ be a corresponding rank-$r$ quantum $k$-colouring of $G$, consisting of projective measurements on $\mathbb{C}^{kr}$ with all projectors being of rank $r$. We may build a $(r,k)$-clump representation of $G$ as follows. For every $v \in V(G), j \in [k]$, take $\{\ket{\psi^v_{i,j}}\}_{i \in [r]}$ to be an orthonormal basis of the support of $E^v_j$. Clearly, for a fixed vertex $v$, these form an orthonormal basis of $\mathbb{C}^{kr}$. Since $\braket{\psi^x_{i,j}|\psi^y_{i',j}} = 0$ for all $(x,y) \in E(G)$, this constitutes a $(r,k)$-clump representation of $G$, as desired.
\end{proof}
\smallskip
Slightly more arduously, in the other direction, we can show:
\begin{theorem} \label{clumps}
For all graphs $G$ and all ranks $r$, it holds that
\[\chi^{(r)}_{q}(G) \leq \xi^{(r)}_c(G)^2\]
\end{theorem}
\begin{proof}
Suppose that $\xi^{(r)}_c(G) = k$, and take $\{\ket{\psi^v_{i, j}}\}$ to be a corresponding $(r,k)$-clump representation of $G$. We build a rank-$r$ quantum $k^2$-colouring of $G$ by appealing to theorem \ref{structurecameron}.  \\

Letting $\zeta$ be a fixed primitive $k$-th root of unity, for every $v \in V(G), i \in [r], c_1,c_2 \in [k]$, define $\ket{\phi^v_{i,(c_1,c_2)}} \in \mathbb{C}^{rk^2}$ by:
\[\ket{\phi^v_{i,(c_1,c_2)}} = \frac{1}{\sqrt{k}} \sum_{j=1}^k \zeta^{c_1 j} \ket{j} \ket{\psi^v_{i, j \oplus c_2}}\]
Where $j \oplus c_2$ is taken to mean the only representative of $j + c_2$ modulo $k$ in $[k]$. For $i,i' \in [r]$ and $c_1,c_2,c'_1,c'_2 \in [k]$, we see that
\begin{align*}
\braket{\phi^v_{i,(c_1,c_2)}|\phi^v_{i',(c'_1,c'_2)}} &= \frac{1}{k} \sum_{j=1}^k \zeta^{(c'_1-c_1)j} \braket{\psi^v_{i, j \oplus c_2} | \psi^v_{i', j \oplus c'_2}}\\
													&= \delta_{i,i'} \delta_{c_2,c'_2} \frac{1}{k} \sum_{j=1}^k \left(\zeta^{(c_1'-c_1)}\right)^{j}\\
													&= \delta_{i,i'} \delta_{c_1,c_1'} \delta_{c_2,c_2'}
\end{align*}
Where the last equality holds because $\zeta^{(c_1'-c_1)}$ is a $k$-th root of unity that is different from one if $c_1 \neq c_1'$. It follows that, for a fixed vertex $v$, the $\ket{\phi^v_{i, (c_1,c_2)}}$ form an orthonormal basis of $\mathbb{C}^{rk^2}$. For every $c_1,c_2 \in [k]$, taking $E^v_{(c_1,c_2)}$ to be the projector on the $r$-dimensional subspace spanned by the $\ket{\phi^v_{i, (c_1,c_2)}}$, we see that for every $v$, these form a projective measurement on $\mathbb{C}^{rk^2}$. Also, we see that for $(u,v) \in V(G)$, for any choice of $i,i' \in [r]$ and for every $c_1,c_2 \in [k]$, by the correctness of the clump representation:
\begin{align*}
\braket{\phi^u_{i,(c_1,c_2)} | \phi^v_{i',(c_1,c_2)}} &= \frac{1}{k} \sum_{j=1}^k \braket{ \psi^u_{i, c_2\oplus j} | \psi^v_{i', c_2 \oplus j}}\\
												 &= 0
\end{align*}
So that $E^u_{(c_1,c_2)} E^v_{(c_1,c_2)} = 0$. The conclusion follows.
\end{proof}

We now put the above result to use to construct the promised graph $G_{21}$. This graph is obtained as the orthogonality graph of the 21 $(2,2)-$clumps listed in appendix A, meaning that two vertices are adjacent if and only if the corresponding clumps are orthogonal. $G_{21}$ is depicted in figure \ref{G21}. The clumps were obtained using a computer by considering the orthogonality graph of all $(2,2)-$clumps with corresponding unit vectors having components in \{$-1$, $0$, $1$\} (pre-normalization) and by looking for a suitable vertex-critical subgraph. \\

By design, we have that $\xi^{(2)}_c(G_{21}) = 2$, and therefore $\chi^{(2)}_q(G_{21}) \leq 4$, by theorem \ref{clumps}. We see that \{1, 3, 10, 16\} forms a clique of $G_{21}$, so that $\omega(G_{21}) \geq 4$: proposition \ref{xisdp} then implies that $\omega(G_{21}) = \chi_q(G_{21}) = \chi^{(2)}_q(G_{21}) = 4$. It is quite easy to show that $\chi(G_{21}) = 5$ using, for example, the standard backtracking algorithm, and it even holds that $\xi(G_{21}) > 4$, which implies that $\xi(G_{21}) = \chi^{(1)}_q(G_{21}) = 5$ and hence proves theorem \ref{separation21}. The description of our computer-assisted proof of this fact is delegated to the next subsection. 

\begin{figure}
\caption{A visual representation of the graph $G_{21}$. This graph contains 21 vertices and 72 edges. The graph6 representation (\cite{graph6}) of $G_{21}$ is \; \; \texttt{TX\_ac\textasciitilde QhaBO\_TDaO@dDewW\_gCd?WWI\_c[?lg} \; \; .}
\label{G21}
\centering
\includegraphics{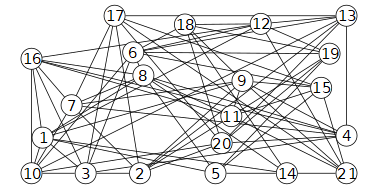}
\end{figure}

\subsection{An algorithm for attempting to prove that a given graph does not admit a $k$-dimensional orthogonal representation}
We now turn to the problem of lower bounding the orthogonal rank of a given graph, with the end goal of showing that $\xi(G_{21}) > 4$. Unlike the case of the chromatic number, which, despite being NP-hard to compute in general, can be computed reasonably comfortably in practice for graphs on up to a few dozen vertices, we know of no practically usable and provably correct algorithm for computing the orthogonal rank of a given graph $G$ of any size except in special cases, like when it so happens that $\ceil{\xi_{SDP}(G)} = \chi(G)$, for example, in which case propositions \ref{sandwich} and \ref{xisdp} imply that $\xi(G) = \chi(G)$. Since $\omega(G_{21}) = \chi_{qc}(G_{21}) = 4$, it follows from this last proposition that $\xi_{SDP}(G_{21}) = 4$, so this not the case here. From the point of view of complexity theory, Briët, Buhrman, Leung, Piovesan and Speelman (\cite{briet}) showed that the problem of determining whether a given graph satisfies $\xi(G) \leq k$ is NP-hard for any fixed $k \geq 3$ (though the bulk of their proof is the $k=3$ case, which can also be seen to follow from lemma 7 of \cite{mancinska}), and the results of Canny (\cite{canny}) imply that this problem is contained in the complexity class PSPACE. It should be mentioned that the problem of looking for a $k$-dimensional orthogonal representation of a given graph can be tackled in practice using numerical root-finding methods, which will generally be reasonably successful at finding one if one exists provided that the graph in question is not too large. Of course, their failure at finding one, while strong evidence that none exists, does not provide an ironclad proof of this, and this is what we are after here.  \\

The problem of determining whether a given graph admits a three-dimensional representation satisfying certain additional properties has received a fair bit of attention in the past, being a subproblem in current approaches in the Kochen-Specker literature, such as those of Arends, Ouaknine and Wampler (\cite{arends}), Uijlen and Westerbaan (\cite{uijlen}), Li, Bright and Ganesh (\cite{li}) and Kirchweger, Peitl and Szeider (\cite{kirchweger}). The purpose of this line of work is to prove lower bounds on the size of any graph that has an orthogonal representation in $\mathbb{R}^3$ but is not 010-colourable (the precise definition of which is not important for our purposes). Though they did not state their results in these terms, the work of Kochen and Specker (\cite{kochenspecker}) can be seen to give the first proof that such a graph exists. Suppose that a given graph has this property. If two vectors in the corresponding orthogonal representation were collinear, identifying the vertices corresponding to these vectors would result in a smaller graph that also has an orthogonal representation in $\mathbb{R}^3$ and is also not 010-colourable: therefore, for the purpose of proving lower bounds, it is permissible to assume that no two vectors in the orthogonal representation are collinear. In the Kochen-Specker literature, a graph is called \textit{embeddable} if it admits such an orthogonal representation. The approach taken in the previously cited works to show that no embeddable and non-010-colourable graph on a certain number of vertices or less exists is analogous to the proof technique that was employed in the first part of this paper: a collection of non-010-colourable graphs is enumerated such that if a counterexample existed, this collection would contain one, and an algorithm due to \cite{uijlen} is employed to prove that none of the graphs in the collection is embeddable. Regrettably, this algorithm appears to be tailor-made for the three-dimensional case and crucially requires the hypothesis that no two vectors in the orthogonal representation be collinear. The requirement that the vectors in the representation must be real is less critical and can be lifted, but doing so will likely increase the complexity of the algorithm.  \\

Fortunately, there is something in this line of work that can be generalised to yield a method of proof for showing that a given graph does not admit a $k$-dimensional orthogonal representation for any choice of $k \geq 3$, namely, the square-free criterion of \cite{arends}, which states that an embeddable graph must be square-free, i.e. not have the complete bipartite graph $K_{2,2}$ as a subgraph. We begin by providing a generalisation of sorts of this, with the original criterion being recovered as the $k=3$ case. Given a graph $G$ and $S \subseteq V(G)$, we define a graph operation called \textit{set identification} in the following way: the result is a supergraph $G'$ of $G$ with the same vertex set such that, for vertices $u,v \in V(G)$ with $u \in S$ and $v \notin S$, the edge $(u,v)$ is added to $G'$ if, for every $w \in S \backslash \{u\}$, it holds that $(v,w) \in E(G)$. We show: 
\smallskip
\begin{theorem} \label{expansion}
Take a graph $G$, $k \geq 3$, and suppose that $v_1,...,v_{k+1}$ are distinct vertices of $G$ such that the corresponding induced subgraph is a supergraph of the complete bipartite graph $K_{k-1,2}$, meaning that for every $i \in [k-1]$ and every $j \in \{k, k+1\}$, it holds that $(v_i, v_j) \in E(G)$. $G$ admits a $k$-dimensional orthogonal representation if and only if one of the following graphs does:
\begin{enumerate}
\item The graph obtained by identifying the vertices $v_k$ and $v_{k+1}$ in $G$. 
\item For every two distinct vertices $u,w \in \{v_1,...,v_{k-1}\}$, the graph obtained by identifying the vertices $u$ and $w$ in $G$. 
\item For every $S \subseteq \{v_1,...,v_{k-1}\}$ with $|S| \geq 3$, the graph obtained by performing set identification with respect to $S$ on $G$. 
\end{enumerate}
Also, a graph coming from one of point 1 and 2 can be dropped if the vertices that are to be identified are adjacent in $G$. Similarly, a graph coming from point 3 corresponding to a set $S$ can be dropped if there is some vertex in $S$ that is adjacent to all the other vertices in $S$.
\end{theorem}
\begin{proof}
The backward direction is straightforward: if one of the graphs $G'$ coming from point 3 has a $k$-dimensional orthogonal representation, so does $G$, since it is a subgraph of $G'$. Similarly, a $k$-dimensional orthogonal representation for a graph coming from one of point 1 or 2 lifts to a $k$-dimensional orthogonal representation of $G$ in the obvious way. \\

For the forward direction, suppose that $G$ has an orthogonal representation $\{\ket{\psi_v}\}_{v \in V(G)}$ in $\mathbb{C}^k$. We proceed by case analysis. One of the following must be the case:
\begin{itemize}
\item $\ket{\psi_{v_1}}, ..., \ket{\psi_{v_{k-1}}}$ are linearly independent. Letting $V$ be the $(k-1)$-dimensional subspace of $\mathbb{C}^{k}$ spanned by these vectors, by basic linear algebra, we have that $V^\bot$ is one-dimensional. Since, by hypothesis, $\ket{\psi_{v_k}}, \ket{\psi_{v_{k+1}}} \in V^\bot$, it follows that these two vectors are collinear. This can be seen to mean that identifying $v_k$ and $v_{k+1}$ in $G$ results in a graph that also has a $k$-dimensional orthogonal representation. Clearly, if $v_k$ and $v_{k+1}$ are adjacent in $G$, this case is impossible.
\item $\ket{\psi_{v_1}}, ..., \ket{\psi_{v_{k-1}}}$ are linearly dependent, so that there exist scalars $\lambda_1, \lambda_2, ..., \lambda_{k-1} \in \mathbb{C}$, not all zero, such that:
\[\sum_{i=1}^{k-1} \lambda_i \ket{\psi_{v_i}} = 0\]
\smallskip
Defining the set $S$ by:
\[S = \{v_i \mid i \in [k-1], \lambda_i \neq 0\}\]
It must be the case that $|S| \geq 2$ because the $\ket{\psi_{v_i}}$ are nonzero. The following cases are possible:
\begin{itemize}
\item $|S| = 2$: letting $S = \{u,w\}$, this case is much the same as the previous one. If $u$ and $w$ are adjacent, this case is impossible and can be omitted, and if not, $\ket{\psi_u}$ and $\ket{\psi_w}$ are collinear, so that by identifying $u$ and $w$ in $G$, we get a graph that also has a $k$-dimensional orthogonal representation. 
\item $|S| \geq 3$: in this case, for every $u \in S$, we see that $\ket{\psi_u}$ can be written as a linear combination of the vectors that were assigned to the other vertices in $S$. If, for some $u \in S$, $u$ is adjacent to all the other vertices in $S$, this case is impossible. If not, pick some $u \in S$. If $v \notin S$ is such that for every $w \in S \backslash \{u\}$, it holds that $(v,w) \in E(G)$, and therefore that $\braket{\psi_v | \psi_w} = 0$, it follows that $\braket{\psi_v|\psi_u} = 0$, by the linearity of the inner product, and so the orthogonal representation remains valid if the edge $(u,v)$ is added to the graph. Hence, applying set identification to $G$ with respect to $S$ results in a graph that also has a $k$-dimensional orthogonal representation. 
\end{itemize}
\end{itemize}
We see that the graphs listed in the statement of the theorem cover all of the above cases. 
\end{proof}

The way the above result can be harnessed to attempt to prove that a given graph $G$ does not admit a $k$-dimensional orthogonal representation is reasonably straightforward. We begin by checking if it is the case that $\xi_{SDP}(G) > k$: if this is so, proposition \ref{xisdp} implies that we are done. If not, we look for a copy of $K_{k-1,2}$ inside $G$ such that the set identifications given by the statement of theorem \ref{expansion} all result in proper supergraphs of $G$. If no such copy exists, no further headway can be made and failure is declared: otherwise, this procedure is applied recursively to the resulting graphs, which are all either smaller than $G$ or contain more edges than it, and, so we hope, are closer to having a $\xi_{SDP}$ value strictly greater than $k$. If the procedure could prove that none of the resulting graphs admits a $k$-dimensional orthogonal representation, we can conclude that neither does $G$. \\

Applied to $G_{21}$ with $k=4$, this algorithm took about three seconds on our hardware to show that $G_{21}$ does not admit a four-dimensional orthogonal representation, with $\xi_{SDP}$ being evaluated on 128 graphs. In our implementation, a copy of $K_{k-1,2}$ inside the graph is picked at random until one is found that is suitable. This choice is surely not optimal, and there are a few variations on theorem \ref{expansion} based on other subgraphs that could be considered. Since the algorithm we described turned out to be perfectly serviceable for our purposes, such considerations are left for future work. It is reasonable that the algorithm worked in this case because $k$ is small, so that the branching factor of the algorithm is not too large (being at most 5), and because it turns out that $G_{21}$ contains a number of distinct copies of $K_{3,2}$, so that the algorithm never got stuck on a graph that was maximal with respect to the operation described by theorem \ref{expansion}. However, it should not be very surprising that the above procedure will not systematically work in all cases, even if it does hold that $G$ does not admit a $k$-dimensional orthogonal representation. Indeed, looking at some of the nonembeddable square-free graphs that were unearthed by works in the Kochen-Specker literature, we could find some for which our algorithm failed at proving the nonexistence of a 3-dimensional orthogonal representation despite numerical evidence strongly suggesting that none exists. \\

Finally, we note that it so happens that for $G_{21}$, $\chi^{(1)}_q$ and $\xi$ coincide, so that tightly lower bounding the latter allowed us to tightly lower bound the former. In general, as was mentioned in the preliminaries, it may happen that these two parameters disagree, in which case our algorithm would be of no direct use for showing a tight lower bound on $\chi^{(1)}_q$. Fortunately, \cite{scarpa} showed that for all graphs $G$ and all values of $k$, $\chi^{(1)}_q(G) \leq k$ if and only if $\xi(G \square K_k) = k$, where $\square$ stands for the Cartesian product of graphs and $K_k$ stands for the complete graph on $k$ vertices. This means that the problem of showing a tight lower bound on the rank-one chromatic number of a given graph can be reduced to the problem of proving that another graph does not admit an orthogonal representation in a given dimension, which could then be tackled using our approach. 

\section{Conclusion and open problems}
In conclusion, we have shown that for every graph $G$ such that either $|V(G)| \leq 13$ or $|V(G)| = 14$ and $\chi_q(G) = 3$, it holds that $\chi(G) = \chi_q(G)$, thereby proving a conjecture of \cite{mancinska} and therefore solving a longstanding open problem about the quantum chromatic number. Furthermore, making use of our notion of a vector clump, we have given a small graph $G_{21}$ on 21 vertices such that $\chi_q(G_{21}) = \chi^{(2)}_q(G_{21}) = 4$ and $\xi(G_{21}) = \chi^{(1)}_q(G_{21}) = \chi(G_{21}) = 5$, thereby giving the smallest separation known between the parameters $\chi_q$ and $\chi^{(1)}_q$, as well as the first separation between $\chi^{(1)}_q$ and  $\chi^{(2)}_q$. \\

Our work suggests a number of avenues for future research, most prominently:
\begin{itemize}
\item Although it would be very surprising if this were not the case, at the time of writing, no proof is known that the $\chi^{(r)}_q$ are all distinct graph parameters, and it could be interesting to look for one. In view of theorem \ref{separation21}, to show this, it would suffice to come up with a graph transformation $G \mapsto G'$ such that $\chi^{(r)}_q(G) = \chi^{(r+1)}_q(G')$ for all $r$, for example. 
\item It would be very interesting to try to use our enumerative approach to gain a better understanding of the quantum chromatic numbers of small graphs. For example, it is reasonable to wonder whether all graphs on 14 vertices exhibiting a separation between the quantum and classical chromatic numbers are subgraphs of $G_{14}$. This could maybe be tackled using our methods, but at the cost of several decades of CPU time at minimum. Alternatively, it would be interesting to look for graphs on 20 vertices or less exhibiting a separation between $\chi_q$ and  $\chi^{(1)}_q$, or, more ambitiously, to try to determine the smallest such graph. Unless, in an unexpected turn of events, such a separation was found on 14 or 15 vertices, it seems unlikely that an enumerative approach like the one that was used in this paper would be practically feasible to attack this last problem.
\item The fact that no general practical algorithm is known for computing the orthogonal rank of a given graph $G$, even if $G$ is small, is a rather problematic state of affairs. It is ironic that although, theoretically, $\chi_q$ is uncomputable and $\xi$ is, in practice and for small graphs, the roles are reversed, with $\chi_q$ being well approachable using semidefinite hierarchies while there are graphs on a dozen vertices whose orthogonal ranks are not rigorously determined. The algorithm we presented in subsection 4.2 did succeed in showing that $\xi(G_{21}) > 4$, but there are graphs for which we will not be so lucky. It seems likely that in future computational attempts to look for small separations between $\chi_q$ and $\chi^{(1)}_q$, a surefire algorithm for computing the orthogonal rank of a graph will be required, although a numerical approach will likely suffice for exploration purposes. 
\end{itemize}
\smallskip
We end by mentioning what we think is an important open problem in the theory of the quantum chromatic number. We know from the work of \cite{BCT} that not only are the classical and quantum chromatic numbers distinct, but that the difference between the two can be arbitrarily large. It would be very interesting to try to strengthen their results by showing that there is some fixed $k$ such that, for every $n \in \mathbb{N}$, there exists a graph $G$ with $\chi_q(G) = k$ and $\chi(G) \geq n$. It should be noted that this cannot be done by only considering quantum colourings of a given fixed rank (as it is possible to upper bound the chromatic number of a graph only knowing its rank-$r$ chromatic number), and except for the construction based on vector clumps that was given in this paper, all known approaches for generating separations between $\chi_q$ and $\chi^{(1)}_q$ rely on variations on the classical 3-SAT to 3-COL reduction and therefore yield graphs with chromatic number at most 4. In particular, in order to establish this, one would have to either put our construction to use or to look for a new way to construct quantum colourings.

\section{Acknowledgments}
Financial support for this work was provided by the Canadian Natural Sciences and Engineering Research Council (NSERC) as well as by the Fonds de recherche du Québec – Nature et technologies (FRQNT). This support is gratefully acknowledged. We are grateful to Brendan McKay for helpful guidance regarding the use of \texttt{geng}. We further thank Harry Buhrman, David Roberson and Ronald de Wolf for useful exchanges about quantum graph theory, Travis Russell for helpful discussions about his hierarchy and William Slofstra and Kieran Mastel for letting us know about the work of \cite{harris}. We also thank Julien Codsi for discussions about the Kochen-Specker literature and for proofreading parts of this manuscript. Finally, we thank our advisors Gilles Brassard and Frédéric Dupuis for their support.

\appendix
\section{The clumps corresponding to the graph $G_{21}$}
This appendix gives the $(2,2)$-clumps of which $G_{21}$ is the orthogonality graph. As in matrix notation, when specifying a clump $\{\ket{\psi_{i,j}}\}$, $i$ runs from top to bottom and $j$ runs from left to right. Note that these clumps are also contained in the \texttt{clumps.jld2} file of \url{https://github.com/lalondeo/QuantumColorings}. 

\begin{multicols}{2}
\begin{enumerate}
\item \begin{align*}
&\frac{1}{\sqrt{2}}  \left[
\begin{array}{cccc}
1 & 0 & -1 & 0 \\
\end{array}
\right] , 
\frac{1}{\sqrt{2}} \left[
\begin{array}{cccc}
1 & 0 & 1 & 0 \\
\end{array}
\right]\\
&\frac{1}{\sqrt{2}} \left[
\begin{array}{cccc}
0 & 1 & 0 & 1 \\
\end{array}
\right] , 
\frac{1}{\sqrt{2}} \left[
\begin{array}{cccc}
0 & 1 & 0 & -1 \\
\end{array}
\right]
\end{align*}
\item \begin{align*}
&\frac{1}{\sqrt{2}} \left[
\begin{array}{cccc}
0 & 0 & 1 & -1 \\
\end{array}
\right] , 
\frac{1}{\sqrt{2}} \left[
\begin{array}{cccc}
1 & -1 & 0 & 0 \\
\end{array}
\right]\\
&\frac{1}{\sqrt{2}} \left[
\begin{array}{cccc}
1 & 1 & 0 & 0 \\
\end{array}
\right] , 
\frac{1}{\sqrt{2}} \left[
\begin{array}{cccc}
0 & 0 & -1 & -1 \\
\end{array}
\right]
\end{align*}
\item \begin{align*}
&\frac{1}{\sqrt{2}} \left[
\begin{array}{cccc}
1 & 0 & 1 & 0 \\
\end{array}
\right] , 
\frac{1}{\sqrt{2}} \left[
\begin{array}{cccc}
-1 & 0 & 1 & 0 \\
\end{array}
\right]\\
&\frac{1}{\sqrt{2}} \left[
\begin{array}{cccc}
0 & 1 & 0 & 1 \\
\end{array}
\right] , 
\frac{1}{\sqrt{2}} \left[
\begin{array}{cccc}
0 & -1 & 0 & 1 \\
\end{array}
\right]
\end{align*}
\item \begin{align*}
&\frac{1}{\sqrt{2}} \left[
\begin{array}{cccc}
1 & 0 & -1 & 0 \\
\end{array}
\right] , 
\frac{1}{\sqrt{2}} \left[
\begin{array}{cccc}
-1 & 0 & -1 & 0 \\
\end{array}
\right]\\
&\frac{1}{\sqrt{2}} \left[
\begin{array}{cccc}
0 & 1 & 0 & 1 \\
\end{array}
\right] , 
\frac{1}{\sqrt{2}} \left[
\begin{array}{cccc}
0 & 1 & 0 & -1 \\
\end{array}
\right]
\end{align*}
\item \begin{align*}
&\frac{1}{\sqrt{2}} \left[
\begin{array}{cccc}
1 & -1 & 0 & 0 \\
\end{array}
\right] , 
\frac{1}{\sqrt{2}} \left[
\begin{array}{cccc}
0 & 0 & -1 & -1 \\
\end{array}
\right]\\
&\frac{1}{\sqrt{2}} \left[
\begin{array}{cccc}
1 & 1 & 0 & 0 \\
\end{array}
\right] , 
\frac{1}{\sqrt{2}} \left[
\begin{array}{cccc}
0 & 0 & -1 & 1 \\
\end{array}
\right]
\end{align*}
\item \begin{align*}
&\frac{1}{\sqrt{2}} \left[
\begin{array}{cccc}
0 & 1 & 1 & 0 \\
\end{array}
\right] , 
\frac{1}{\sqrt{2}} \left[
\begin{array}{cccc}
0 & 1 & -1 & 0 \\
\end{array}
\right]\\
&\frac{1}{\sqrt{2}} \left[
\begin{array}{cccc}
1 & 0 & 0 & 1 \\
\end{array}
\right] , 
\frac{1}{\sqrt{2}} \left[
\begin{array}{cccc}
1 & 0 & 0 & -1 \\
\end{array}
\right]
\end{align*}
\item \begin{align*}
&\frac{1}{\sqrt{2}} \left[
\begin{array}{cccc}
1 & 0 & -1 & 0 \\
\end{array}
\right] , 
\frac{1}{\sqrt{2}} \left[
\begin{array}{cccc}
1 & 0 & 1 & 0 \\
\end{array}
\right]\\
&\frac{1}{\sqrt{2}} \left[
\begin{array}{cccc}
0 & 1 & 0 & 1 \\
\end{array}
\right] , 
\frac{1}{\sqrt{2}} \left[
\begin{array}{cccc}
0 & -1 & 0 & 1 \\
\end{array}
\right]
\end{align*}
\item \begin{align*}
&\frac{1}{2} \left[
\begin{array}{cccc}
1 & 1 & 1 & -1 \\
\end{array}
\right] , 
\frac{1}{2} \left[
\begin{array}{cccc}
-1 & -1 & 1 & -1 \\
\end{array}
\right]\\
&\frac{1}{2} \left[
\begin{array}{cccc}
1 & -1 & 1 & 1 \\
\end{array}
\right] , 
\frac{1}{2} \left[
\begin{array}{cccc}
-1 & 1 & 1 & 1 \\
\end{array}
\right]
\end{align*}
\item \begin{align*}
&\frac{1}{2} \left[
\begin{array}{cccc}
1 & 1 & 1 & -1 \\
\end{array}
\right] , 
\frac{1}{2} \left[
\begin{array}{cccc}
-1 & 1 & 1 & 1 \\
\end{array}
\right]\\
&\frac{1}{2} \left[
\begin{array}{cccc}
1 & -1 & 1 & 1 \\
\end{array}
\right] , 
\frac{1}{2} \left[
\begin{array}{cccc}
1 & 1 & -1 & 1 \\
\end{array}
\right]
\end{align*}
\item \begin{align*}
&\frac{1}{\sqrt{2}} \left[
\begin{array}{cccc}
0 & 1 & 0 & -1 \\
\end{array}
\right] , 
\frac{1}{\sqrt{2}} \left[
\begin{array}{cccc}
0 & -1 & 0 & -1 \\
\end{array}
\right]\\
&\frac{1}{\sqrt{2}} \left[
\begin{array}{cccc}
1 & 0 & -1 & 0 \\
\end{array}
\right] , 
\frac{1}{\sqrt{2}} \left[
\begin{array}{cccc}
-1 & 0 & -1 & 0 \\
\end{array}
\right]
\end{align*}
\item \begin{align*}
&\frac{1}{\sqrt{2}} \left[
\begin{array}{cccc}
1 & 0 & -1 & 0 \\
\end{array}
\right] , 
\frac{1}{\sqrt{2}} \left[
\begin{array}{cccc}
0 & -1 & 0 & 1 \\
\end{array}
\right]\\
&\frac{1}{\sqrt{2}} \left[
\begin{array}{cccc}
0 & 1 & 0 & 1 \\
\end{array}
\right] , 
\frac{1}{\sqrt{2}} \left[
\begin{array}{cccc}
1 & 0 & 1 & 0 \\
\end{array}
\right]
\end{align*}
\item \begin{align*}
&\frac{1}{2} \left[
\begin{array}{cccc}
1 & -1 & -1 & 1 \\
\end{array}
\right] , 
\frac{1}{2} \left[
\begin{array}{cccc}
-1 & 1 & -1 & 1 \\
\end{array}
\right]\\
&\frac{1}{2} \left[
\begin{array}{cccc}
1 & 1 & 1 & 1 \\
\end{array}
\right] , 
\frac{1}{2} \left[
\begin{array}{cccc}
-1 & -1 & 1 & 1 \\
\end{array}
\right]
\end{align*}
\item \begin{align*}
&\frac{1}{2} \left[
\begin{array}{cccc}
1 & 1 & -1 & -1 \\
\end{array}
\right] , 
\frac{1}{2} \left[
\begin{array}{cccc}
1 & 1 & 1 & 1 \\
\end{array}
\right]\\
&\frac{1}{2} \left[
\begin{array}{cccc}
1 & -1 & 1 & -1 \\
\end{array}
\right] , 
\frac{1}{2} \left[
\begin{array}{cccc}
-1 & 1 & 1 & -1 \\
\end{array}
\right]
\end{align*}
\item \begin{align*}
&\frac{1}{\sqrt{2}} \left[
\begin{array}{cccc}
0 & 0 & 1 & -1 \\
\end{array}
\right] , 
\frac{1}{\sqrt{2}} \left[
\begin{array}{cccc}
1 & 1 & 0 & 0 \\
\end{array}
\right]\\
&\frac{1}{\sqrt{2}} \left[
\begin{array}{cccc}
0 & 0 & 1 & 1 \\
\end{array}
\right] , 
\frac{1}{\sqrt{2}} \left[
\begin{array}{cccc}
1 & -1 & 0 & 0 \\
\end{array}
\right]
\end{align*}
\item \begin{align*}
& \left[
\begin{array}{cccc}
0 & 0 & 1 & 0 \\
\end{array}
\right] , 
 \left[
\begin{array}{cccc}
0 & -1 & 0 & 0 \\
\end{array}
\right]\\
&\left[
\begin{array}{cccc}
0 & 0 & 0 & 1 \\
\end{array}
\right] , 
\left[
\begin{array}{cccc}
-1 & 0 & 0 & 0 \\
\end{array}
\right]
\end{align*}
\item \begin{align*}
&\frac{1}{\sqrt{2}} \left[
\begin{array}{cccc}
0 & 1 & 0 & -1 \\
\end{array}
\right] , 
\frac{1}{\sqrt{2}} \left[
\begin{array}{cccc}
0 & 1 & 0 & 1 \\
\end{array}
\right]\\
&\frac{1}{\sqrt{2}} \left[
\begin{array}{cccc}
1 & 0 & 1 & 0 \\
\end{array}
\right] , 
\frac{1}{\sqrt{2}} \left[
\begin{array}{cccc}
1 & 0 & -1 & 0 \\
\end{array}
\right]
\end{align*}
\item \begin{align*}
&\frac{1}{2} \left[
\begin{array}{cccc}
1 & -1 & -1 & -1 \\
\end{array}
\right] , 
\frac{1}{2} \left[
\begin{array}{cccc}
-1 & -1 & -1 & 1 \\
\end{array}
\right]\\
&\frac{1}{2} \left[
\begin{array}{cccc}
1 & -1 & 1 & 1 \\
\end{array}
\right] , 
\frac{1}{2} \left[
\begin{array}{cccc}
1 & 1 & -1 & 1 \\
\end{array}
\right]
\end{align*}
\item \begin{align*}
&\frac{1}{2} \left[
\begin{array}{cccc}
1 &  1 & -1 & -1 \\
\end{array}
\right] , 
\frac{1}{2} \left[
\begin{array}{cccc}
1 & -1 & -1 & 1 \\
\end{array}
\right]\\
&\frac{1}{2} \left[
\begin{array}{cccc}
1 & -1 & 1 & -1 \\
\end{array}
\right] , 
\frac{1}{2} \left[
\begin{array}{cccc}
-1 & -1 & -1 & -1 \\
\end{array}
\right]
\end{align*}
\item \begin{align*}
&\frac{1}{2} \left[
\begin{array}{cccc}
1 & 1 & -1 & -1 \\
\end{array}
\right] , 
\frac{1}{2} \left[
\begin{array}{cccc}
-1 & 1 & 1 & -1 \\
\end{array}
\right]\\
&\frac{1}{2} \left[
\begin{array}{cccc}
1 & -1 & 1 & -1 \\
\end{array}
\right] , 
\frac{1}{2} \left[
\begin{array}{cccc}
1 & 1 & 1 & 1 \\
\end{array}
\right]
\end{align*}
\item \begin{align*}
&\frac{1}{\sqrt{2}} \left[
\begin{array}{cccc}
0 & 1 & 1 & 0 \\
\end{array}
\right] , 
\frac{1}{\sqrt{2}} \left[
\begin{array}{cccc}
-1 & 0 & 0 & 1 \\
\end{array}
\right]\\
&\frac{1}{\sqrt{2}} \left[
\begin{array}{cccc}
1 & 0 & 0 & 1 \\
\end{array}
\right] , 
\frac{1}{\sqrt{2}} \left[
\begin{array}{cccc}
0 & -1 & 1 & 0 \\
\end{array}
\right]
\end{align*}
\item \begin{align*}
&\left[
\begin{array}{cccc}
1 & 0 & 0 & 0 \\
\end{array}
\right] , 
\left[
\begin{array}{cccc}
0 & 0 & 0 & -1 \\
\end{array}
\right]\\
&\left[
\begin{array}{cccc}
0 & 1 & 0 & 0 \\
\end{array}
\right] , 
\left[
\begin{array}{cccc}
0 & 0 & -1 & 0 \\
\end{array}
\right]
\end{align*}
\end{enumerate}
\end{multicols}
\end{document}